\documentclass[conf]{new-aiaa}
\usepackage[utf8]{inputenc}
\usepackage{xcolor}

\usepackage{graphicx}
\usepackage{cleveref}
\usepackage{booktabs}
\usepackage{multirow}
\usepackage{subfigure}
\usepackage{siunitx}
\usepackage{mathtools}
\usepackage{longtable,tabularx}
\newcommand{\X}{\mathbf{x}}
\newcommand{\U}{\mathbf{u}}
\newcommand{\R}{\mathbb{R}}
\newcommand{\Xss}{\mathcal X_{\operatorname{safe}}^*}
\newcommand{\Xs}{\mathcal X_{\operatorname{safe}}}
\newcommand{\Xa}{\mathcal X_{\operatorname{avoid}}}
\newcommand{\itt}[1]{\{0, 1, \ldots #1 \}}
\newcommand{\ittinfty}{\mathbb{W}}
\newcommand{\hdc}{d}
\newcommand{\hcbf}{h}

\newcommand{\mcbf}{M_{\operatorname{CBF}}}
\newcommand{\oura}{MPC-MCI}
\newcommand{\Rnmpc}{\mathcal R^{\operatorname{NMPC-DCBF}}}
\newcommand{\Roura}{\mathcal R^{\operatorname{\oura}}}
\newcommand{\decSpace}{\vspace{-3pt}}

\usepackage{amsthm}
\newtheorem{theorem}{Theorem}
\newtheorem{problem}{Problem}
\newtheorem{lemma}{Lemma}

\theoremstyle{definition}
\newtheorem{defn}{Definition}
\theoremstyle{remark}
\newtheorem{rmrk}{Remark}
\newtheorem{asmp}{Assumption}
\newtheorem{corollary}{Corollary}[theorem]

\Crefname{lemma}{Lemma}{Lemmas}
\Crefname{problem}{Problem}{Problem}
\Crefname{theorem}{Theorem}{Theorems}
\Crefname{defn}{Definition}{Definition}
\Crefname{rmrk}{Remark}{Remark}
\Crefname{asmp}{Assumption}{Assumptions}
\Crefname{prop}{Proposition}{Proposition}
\Crefname{corollary}{Corollary}{Corollary}
\crefname{example}{Example}{Example}
\crefname{equation}{Eq.\!}{Eqs.\!}
\crefname{figure}{Fig.\!}{Figs.\!}

\title{Maximal Controlled Invariant-MPC: Enhancing Feasibility and Reducing Conservatism through Terminal CBF Constraint in Safety-Critical Control}

\author{Tanmay Dokania\footnote{Graduate Student, Daniel Guggenheim School of Aerospace Engineering.} and Yashwanth Kumar Nakka\footnote{Assistant Professor, Director, Aerospace Robotics Lab, Daniel Guggenheim School of Aerospace Engineering.}}
\affil{Georgia Institute of Technology, Atlanta, GA, 30332}

\begin{document}

\maketitle

\begin{abstract}
Optimal control for safety-critical systems is often dependent on the conservativeness of constraints. 
Control Barrier Functions (CBFs) serve as a medium to represent such constraints, but constructing a minimally conservative CBF is a computationally intractable problem. Therefore, approaches that can guarantee safety while reducing conservatism will help improve the optimality of the system under consideration.
Here, we present a Model Predictive Control (MPC) formulation using CBF as a terminal constraint, which is proven to improve feasibility and reachable sets with increasing prediction horizon. The constructive nature of the proofs allows for warm-starting the nonlinear optimization problem, thereby reducing the computational time substantially. Simulations are set up for a simple nonholonomic system to numerically validate the results, and it is observed that the number of infeasible points decreased by a factor of $1.7$ to $2.7$. The increase in reachable state space was demonstrated by the ability of the system to track trajectories that are entirely inside the unsafe region of the control barrier function.
\end{abstract}

\section{Nomenclature}

{\renewcommand\arraystretch{1.0}
\noindent\begin{longtable*}{@{}l @{\quad=\quad} l@{}}
$\X_t$          & system state at time $t$ \\
$\U_t$          & control input at time $t$ \\
$\mathcal{X}$   & state space \\
$\mathcal{U}$   & admissible control set \\
$\Xs$           & known safe set (0-superlevel set of $\hcbf$) \\
$\Xa$           & known unsafe (avoid) set \\
$\Xss$          & maximal controlled invariant subset of $\mathcal{X} \setminus \Xa$ \\
$\mathcal{X}_f$ & terminal constraint set \\
$f$             & discrete-time system dynamics \\
$\hcbf$         & discrete-time control barrier function \\
$\hdc$          & signed distance function from obstacles \\
$\phi_f^\pi$    & flow (state trajectory) under dynamics $f$ and control policy $\pi$ \\
$\pi$           & control policy/trajectory \\
$N$             & prediction horizon \\
$\mcbf$         & number of CBF-constrained steps in NMPC-DCBF \\
$\gamma$        & CBF decay rate parameter \\
$\omega_k$      & slack variable at step $k$ \\
$J_t^*$         & optimal cost-to-go at time $t$ \\
$p(\cdot)$      & terminal cost function \\
$q(\cdot,\cdot)$& stage cost function \\
$\psi(\cdot)$   & slack variable cost function \\
$\Rnmpc_{N,M}$  & one-step reachable set for NMPC-DCBF \\
$\Roura_N$      & one-step reachable set for MPC-MCI \\
$\Delta t$      & sampling time, s \\
$a_m$           & linear acceleration bound, m/s$^2$ \\
$\alpha_m$      & angular acceleration bound, rad/s$^2$ \\
$r_0$           & obstacle radius, m \\
$v$             & linear velocity, m/s \\
$\omega$        & angular velocity, rad/s \\
$\theta$        & heading angle, rad \\
$\mathbb{W}$    & set of whole numbers \\
$\mathbb{R}$    & set of real numbers \\
$\itt{k}$       & index set $\{0, 1, \ldots, k\}$
\end{longtable*}}
\section{Introduction}
Providing formal safety guarantees for safety-critical systems is a prominent problem in robotics and control theory. Such guarantees are often given using CBFs with two properties. First, they are negative inside the unsafe region. Second, they satisfy the derivative ascent condition, which states the existence of a control input satisfying a bound on its rate of change. It is easy to construct functions that satisfy the first property. However, satisfying the second property requires deep intuition about system dynamics or computationally intensive approaches like HJB reachability analysis, which becomes intractable for higher-order systems \cite{bansal2017hamilton}. It is more often possible to construct overly conservative CBFs that restrict the system to a much smaller set of allowed states, where a heuristic can lead to the construction of a CBF. Potentially conflicting safety constraints and optimal control conditions require one to construct the least constraining safety constraints to reduce the suboptimality.

\subsection{Related Work}
CBFs have been used to encode input constraints for set invariance, with discrete-time versions proposed in \cite{agrawal2017discrete}, but without longer horizon optimization, the system reacts myopically to obstacles. Initial MPC formulations in \cite{zeng2021safety} used candidate functions lacking invariance to show how decay rates improve performance near unsafe regions, while \cite{zeng2021enhancing} introduced slack variables to enhance feasibility, though without guarantees of recursive feasibility.
Ref.~\cite{schilliger2021control} presents a tube-based MPC with a terminal Discrete Time-CBF (DT-CBF) that handles integration errors, but it assumes initial feasibility and certificate--constraint compatibility. Ref.~\cite{wabersich2022predictive} instead proposes a safety filter with a soft-constrained terminal DT-CBF to recover from infeasibility.
Ref.~\cite{katriniok2023discrete} employs quasi-CBFs and CBFs to design two-step terminal constraints and a transient function constraint that are not required to be control invariant, offering both recursive feasibility and safety guarantees. However, it does not provide formal proofs for reducing the conservatism introduced by quasi-CBFs, with the main emphasis placed on recursive feasibility. The analysis is further limited to a linear setting involving two double integrators.

\begin{figure}[hbt!]
    \centering
    \includegraphics[width=0.5\textwidth]{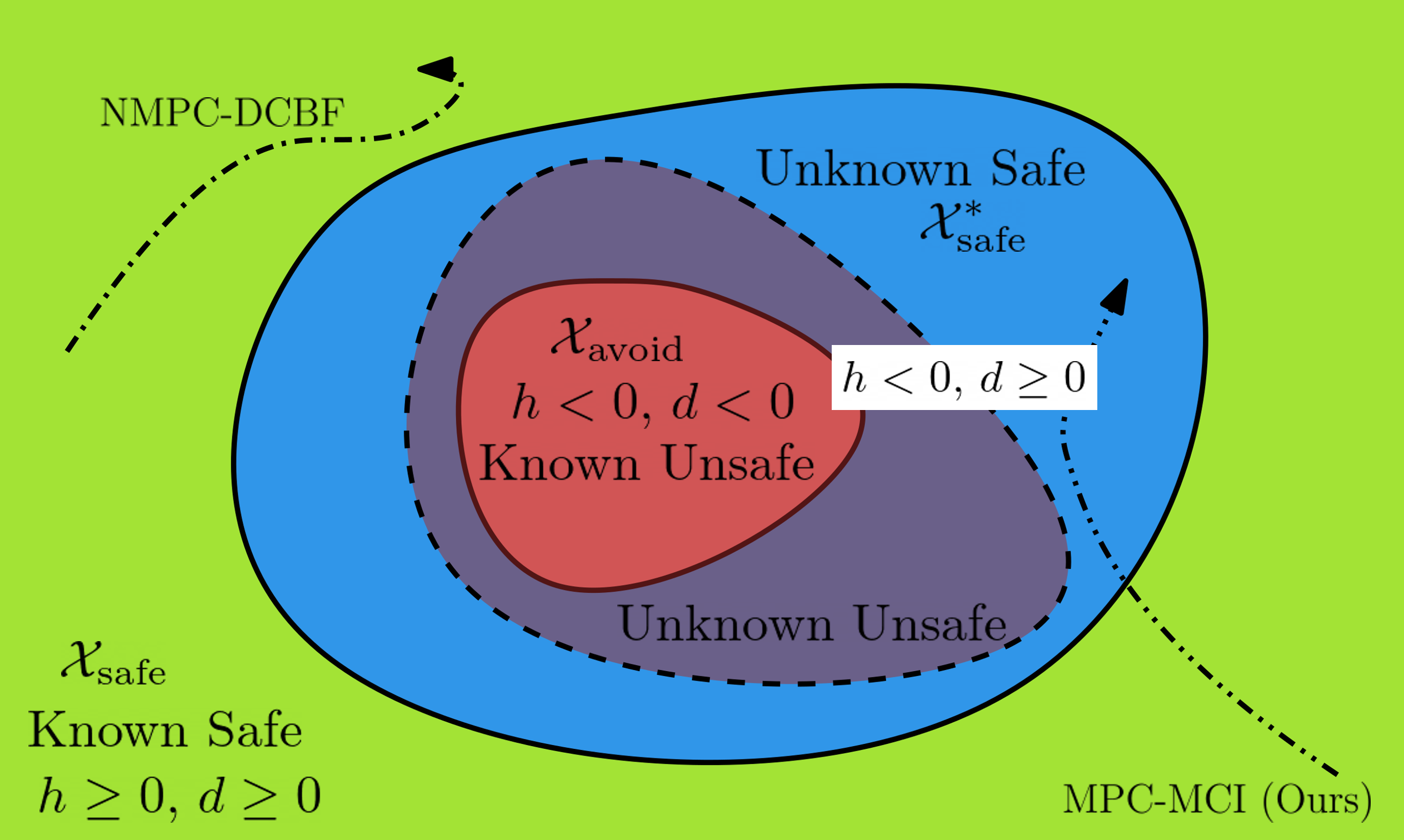}
    \caption{For a given obstacle configuration, i.e., the set of known unsafe states (denoted by $\Xa$), a conservative CBF can be constructed, which yields the safe set $\Xs$. The dashed boundary represents an unknown boundary defined by the maximal controlled invariant set $\Xss$. The earlier method (NMPC-DCBF) is overly conservative because it remains within the known safe region; however, the proposed method (MPC-MCI) exits the known safe region while staying within the unknown safe region, thereby reducing conservatism and enhancing feasibility.}
    \label{fig:cartoon_safety}
\end{figure}

\subsection{Contributions}
The main contributions of the paper are as follows:
\begin{itemize}
    \item Develops an optimization-based MPC framework that is recursively feasible and guarantees safety by keeping the system within the maximal controlled invariant safe subset; such guarantees are absent in \cite{zeng2021safety, zeng2021enhancing}
    \item Requires only initial feasibility to guarantee safety for the future states, unlike being in the conservative safe sets imposed by \cite{zeng2021safety, zeng2021enhancing}
    \item Imposes a single-step CBF terminal constraint, as opposed to two-step constraint in \cite{katriniok2023discrete}
    \item Presents constructive proof for recursive feasibility that allows warm-starting the nonlinear optimization, thereby speeding up the computation
    \item Uses reachability analysis to demonstrate that the feasible set and the reachable set expand with increasing prediction horizon
    \item Numerically validates the enhancement of the feasible set and the reachability of the proposed model predictive control for a 2D nonholonomic system through simulations
\end{itemize}

\subsection{Paper Structure}
The paper is organized as follows: Sec.~\ref{sec:background} introduces relevant background on discrete dynamics, controlled invariance, CBFs, and MPC approaches. Sec.~\ref{sec:main} presents our MPC-MCI algorithm with theoretical reachability enhancement guarantees. Sec.~\ref{sec:sim} demonstrates feasibility and tracking via simulations of a nonholonomic unicycle. Finally, Sec.~\ref{sec:conc} concludes the paper.

\section{Background}\label{sec:background}
In this section, we present the notation, preliminaries on discrete dynamics, controlled invariant sets, discrete-time CBFs, and the optimal control formulation for model-predictive control.

\subsection{Discrete-Time Nonlinear Control System}
We consider a discrete-time nonlinear control system of the form
\begin{equation}
    \X_{t+1} = f(\X_t, \U_t), \quad \X_t \in \mathcal X,\ \U_t \in \mathcal U,
    \label{eq:dis_dynamics}
\end{equation}
where $\X_t$ denotes the system state at time $t$, $\U_t$ is the control input, $\mathcal X$ is the state space, and $\mathcal U$ is the admissible control set.
For a given control trajectory $\pi: \itt{T} \to \mathcal U$ and the dynamics $f$, the flow $\phi_f^\pi: \itt{T+1} \times \mathcal X\to \mathcal X$ is defined as:
\begin{equation}
\begin{split}
    \phi_f^\pi(0, \X) &= \X,\\ \phi_f^\pi(t+1, \X) &= f(\phi_f^\pi(t, \X), \pi(t)) \ \forall t\in \itt{T}.
\end{split}
\end{equation}

\begin{asmp}
    The system dynamics $f$ in Eq.~\eqref{eq:dis_dynamics} are perfectly known, and the model is free from external disturbances or parametric uncertainties.
    \label{a:model}
\end{asmp}

\begin{asmp}
    At each time step $t$, the full state information $\X_t$ is assumed to be available for state feedback control design. \label{a:state-info}
\end{asmp}

\subsection{Controlled Invariance and Control Barrier Functions}
In safety-critical control tasks like obstacle avoidance, system trajectories must remain within a set of safe states indefinitely.

\begin{defn}[Controlled Invariant Set \cite{ames2019control}]
    A set $\mathcal P \subseteq \mathcal X$ is called \emph{control invariant} if, for every state $\X \in \mathcal P$, there exists a control input $\U \in \mathcal U$ such that $f(\X, \U) \in \mathcal P$.
    \label{def:cnt-inv}
\end{defn}

In particular, if $\mathcal P$ is controlled invariant, then for any $\X \in \mathcal P$ one can construct a control policy $\pi:\mathbb{W} \to \mathcal U$ such that the resulting trajectory satisfies $\phi_f^\pi(t,\X) \in \mathcal P$ for all $t \in \mathbb{W}$. Here, $\mathbb W$ and $\mathbb R$ denote the set of whole and real numbers, respectively.
Thus, being inside a controlled invariant set free of obstacles provides a natural notion of \emph{safety}.

\begin{defn}[Discrete-Time Control Barrier Function \cite{ames2019control}]\label{def:cbf}
    Let $\Xs$ denote the $0$-superlevel set of a function $h:\mathcal X \to \mathbb R$.
    The function $h$ is called a \emph{discrete-time control barrier function (CBF)} on $\Xs \subseteq \mathcal X$ for the dynamics in Eq.~\eqref{eq:dis_dynamics} if, for all $\X \in \Xs$, there exists a $\gamma \in (0,1]$ and a control input $\U \in \mathcal U$ such that
    \begin{equation}
        \Delta h(\X,\U) \coloneqq h(f(\X,\U)) - h(\X) \; \ge \; -\gamma h(\X).
    \end{equation}
\end{defn}

In navigation or obstacle-rich environments, unsafe states can often be described using a signed distance function.
Let $\hdc:\mathcal X \to \mathbb R$ denote such a function, and define the \emph{unsafe set} as
\begin{equation}
    \Xa \coloneqq \{\X \in \mathcal X \mid \hdc(\X) < 0\}.
\end{equation}
Here, $\Xa$ represents states that correspond to collisions or constraint violations.

Conversely, suppose there exists a discrete-time control barrier function $\hcbf:\mathcal X \to \mathbb R$ associated with $\Xa$.
Then the corresponding \emph{safe set} is given by
\begin{equation}
    \Xs \coloneqq \{\X \in \mathcal X \mid \hcbf(\X) \ge 0\},
\end{equation}
which can be shown to be a controlled invariant set.
Thus, $\Xs$ can be regarded as the \textbf{known safe} set.
In the context of navigation, $\Xs$ represents states that are either already away from obstacles, or can be actively driven away from them through an appropriate control signal.
Since both the CBF $\hcbf$ and the signed distance function $\hdc$ are constructed for the same obstacle configuration, it follows that
\[
    \Xs \subseteq \mathcal X \setminus \Xa,
\]
which under mild Lipschitz regularity of $\hcbf$ and $\hdc$ yields the compatibility condition
\begin{equation}
    \hcbf(\X) \ge 0 \;\; \implies \;\; \hdc(\X) \ge 0.
    \label{eq:CBFtoDC}
\end{equation}
The same has been depicted in Fig.~\ref{fig:cartoon_safety}.

\begin{rmrk}
    The relation in Eq.~\eqref{eq:CBFtoDC} should be interpreted as a \emph{consistency condition} between the chosen CBF and the signed distance function.
    It ensures that any state certified as safe by the CBF is also evaluated as safe according to the geometric notion of distance from obstacles.
    This condition does not introduce new restrictions but reflects the fact that both functions encode the same obstacle geometry.
\end{rmrk}

The region $\mathcal X \setminus (\Xa \cup \Xs)$ consists of states from which the system may either recover to the \textbf{known safe} set $\Xs$, or alternatively, evolve into the \textbf{known unsafe} set $\Xa$, depending on the control applied.
To characterize the subset of such states that can eventually be steered into $\Xs$, we introduce the following definition.

\begin{defn}[Maximal Controlled Invariant Subset]\label{def:mcnt-inv}
    Given a set $\mathcal P_1 \subseteq \mathcal X$, a set $\mathcal P_2 \subseteq \mathcal P_1$ is called the \emph{maximal controlled invariant subset} of $\mathcal P_1$ under the dynamics $f$ in Eq.~\eqref{eq:dis_dynamics} if, for every $\X \in \mathcal P_2$, there exists a control policy $\pi:\ittinfty \to \mathcal U$ such that
    \begin{equation}
        \phi_f^\pi(t,\X) \in \mathcal P_1, \quad \forall t \in \ittinfty.
    \end{equation}
    Equivalently,
    \begin{equation}
        \mathcal P_2
        = \big\{ \X \in \mathcal P_1 \;\big|\; \exists\ \pi:\ittinfty \to \mathcal U,\
        \phi_f^\pi(t,\X) \in \mathcal P_1,\ \forall t \in \ittinfty \big\}.
    \end{equation}
\end{defn}

Let $\Xss$ denote the maximal controlled invariant subset of $\mathcal X \setminus \Xa$.
Intuitively, $\Xss$ is the largest set of states from which the system can be kept safe (i.e., outside the unsafe set $\Xa$) for all future time through appropriate control actions.
Although computing $\Xss$ exactly is generally intractable, the relation
\[
    \Xs \subseteq \Xss
\]
follows directly from the definitions, since any state certified as safe by the CBF is necessarily a member of the maximal invariant set.

In the following sections, we develop a model predictive control (MPC) scheme that enforces invariance with respect to $\Xss$ rather than only $\Xs$.
This shift allows us to guarantee safety while simultaneously reducing the conservatism introduced by constraining the trajectories solely to the known safe set $\Xs$.

\subsection{Overview of Recent MPC-CBF Approaches}
In this section, we revisit existing optimal control formulations that integrate DT-CBFs and distance functions into model predictive control (MPC).
The main limitation of prior methods is that they either enforce safety through geometric distance constraints (MPC-DC), restrict trajectories to the conservative safe set $\Xs$ defined by CBFs (MPC-DCBF, NMPC-DCBF), or impose invariance through two-step quasi-CBFs (DTCBF-MPC).
In contrast, our proposed approach develops an MPC formulation that enforces invariance with respect to the maximal controlled invariant set $\Xss$ with a single terminal constraint.
This shift reduces conservatism while still guaranteeing safety, since $\Xs \subseteq \Xss$ by construction.

\subsubsection{Model Predictive Control.}\label{sec:mpc-vanilla}
Using the dynamics as constraints, one can formulate a finite-horizon optimization which is solved at each time step to obtain the control input. The decision variables of the optimization problem are $X, U$, which are defined as:
\begin{equation}
\begin{split}
U &:= [\U_{t|t}^\top, \U_{t+1|t}^\top, \ldots,  \U_{t+N-1|t}^\top ]^\top,\\
X &:=[\X_{t+1|t}^\top, \X_{t+2|t}^\top, \ldots, \X_{t+N|t}^\top]^\top.
\end{split}\label{eq:decision-variable-def}
\end{equation}
Here, $\U_{t+k|t}$ and $\X_{t+k|t}$ denote the predicted control input and state for step $t+k$ for the prediction made at step $t$, respectively. The optimization problem is formulated as:

\begin{problem}[\textbf{MPC}]
\decSpace
\label{prob:mpc-vanilla}
\begin{align}
    J_{t}^{*}(\X_t) =  \min_{U,X} p(\X_{t+N|t}){+}\sum_{k=0}^{N-1} & q(\X_{t+k|t},\U_{t+k|t}) \label{eq:mpc-cost}\\
    \operatorname{s.t.} \quad
    \X_{t+k+1|t} = f(\X_{t+k|t}, \U_{t+k|t}), & \ k \in\itt{N-1}  \label{eq:mpc-system-dynamics}\\
    \X_{t+k|t} \in \mathcal{X}, \U_{t+k|t} \in \mathcal{U}, & \ k \in\itt{N-1} \label{eq:mpc-state-input-constraint}\\
    \X_{t|t} = \X_t, \label{eq:mpc-current-state} \\
    \X_{t+N|t} \in \mathcal{X}_f. & \label{eq:mpc-terminal-set}
\end{align}
\end{problem}

Here, at time $t$, $J^*_t$ denotes the optimized cost-to-go and $N$ denotes the prediction horizon. The terms $p(\X_{t+N|t})$ and $q(\X_{t+k|t},\U_{t+k|t})$ in Eq.~\eqref{eq:mpc-cost} denote the terminal cost and the stage cost at step $t+k$ respectively. The constraint~\eqref{eq:mpc-system-dynamics} describes the system dynamics, and Eq.~\eqref{eq:mpc-state-input-constraint} describes the state and input constraints for each step. The constraints~\eqref{eq:mpc-current-state} and~\eqref{eq:mpc-terminal-set} impose the initial and final state conditions on the optimization problem.

The optimal solution to Problem~\ref{prob:mpc-vanilla} at time $t$ is a sequence of control inputs $U^* = [\U_{t|t}^{*\top}, \U_{t+1|t}^{*\top}, \ldots,  \U_{t+N-1|t}^{*\top} ]^\top$, whose first element is applied to the system, i.e., $\U_t = \U_{t|t}^*$, and then Problem~\ref{prob:mpc-vanilla} is reformulated for the next time step $t=t+1$, with $\X_{t+1|t}$ as the initial state, resulting in a receding horizon control strategy.

\subsubsection{NMPC-DCBF~\cite{zeng2021enhancing}.}
The decay rates $\gamma_k$ are fixed in the CBF constraint in MPC-DCBF \cite{zeng2021safety}, which are relaxed by introducing additional slack variables denoted by:
\begin{equation}
    \Omega := [\omega_0, \omega_1, \ldots, \omega_{\mcbf-1}]^\top.
\end{equation}
This method is proposed in Ref.~\cite{zeng2021enhancing}:

\begin{problem}[\textbf{NMPC-DCBF}]
\decSpace
\label{prob:cbf-nmpc}
\begin{align}
    J_{t}^{*}(\X_t){=}\min_{U, \Omega}\ & \beta V(\X_{t+N|t}){+}\sum_{k=0}^{N-1}q(\X_{t+k|t},\U_{t+k|t}){+}\psi(\omega_k) \label{subeq:cbf-nmpc-cost}\\
    \operatorname{s.t.}\quad &   \eqref{eq:mpc-system-dynamics},~\eqref{eq:mpc-state-input-constraint},~\eqref{eq:mpc-current-state},~\eqref{eq:mpc-terminal-set}, \notag\\
    &\hcbf (\X_{t+k+1|t}) \geq  \omega_k(1 -  \gamma_k) \hcbf(\X_{t+k|t}),  \notag  \\ &\omega_k \geq 0, \ k\in\itt{\mcbf{-}1}. \label{subeq:cbf-nmpc-cbf-constraint}
\end{align}
\end{problem}

\noindent
Here, the CBF constraint is set for the first $\mcbf \le N$ steps. Therefore, it leads to an increase in the number of optimization variables by $\mcbf$ and also requires additional tuning of defining a cost function $\psi$ for the slack variables. This increases the computation complexity, which, however, can be decreased by decreasing $\mcbf$.
In Ref.~\cite{zeng2021enhancing}, this approach has been demonstrated to increase the feasibility and safety as opposed to other methods like MPC-DCBF \cite{zeng2021safety} and MPC-GCBF \cite{ma2021feasibility}.

\subsubsection{DTCBF-MPC \cite{katriniok2023discrete}.}
The CBF constraint has been imposed for several states in the approaches so far; however, this approach proposes to generalize this by allowing a new constraint function $H$.

\begin{problem}[\textbf{DTCBF-MPC}]
\decSpace
\label{prob:dtcbf-mpc}
\begin{align}
    J_{t}^{*}(\X_t){=}\min_{U}\ & J(\X_t,U) \label{subeq:dtcbf-mpc-cost}\\
    \operatorname{s.t.}\quad &   \eqref{eq:mpc-system-dynamics},~\eqref{eq:mpc-state-input-constraint},~\eqref{eq:mpc-current-state},~\eqref{eq:mpc-terminal-set}, \notag\\
    & H(\X_{t+k|t})\ge 0,\ k \in \{1,2, \ldots,N-2\} \label{eq:dtcbf-H}\\
    & \hcbf_2(\X_{t+N-1|t}) \ge 0 \label{eq:dtcbf-h1}\\
    &\hcbf_2 (\X_{t+N|t}) \geq  (1 -  \gamma_k) \hcbf_2(\X_{t+N-1|t}). \label{eq:dtcbf-h2}
\end{align}
\end{problem}

\noindent
The function $\hcbf_2$ is a quasi-CBF, which is a two-step generalization of CBFs. This constraint needs to be applied on the last two time steps to ensure forward invariance, Eqs.~\eqref{eq:dtcbf-h1} and~\eqref{eq:dtcbf-h2}.
The constraint in Eq.~\eqref{eq:dtcbf-H} denotes a general constraint function that satisfies $H(\X)\ge\hcbf_2(\X)\ge0$. This leads to a relaxation of the quasi-CBF constraint imposed.
However, Ref.~\cite{katriniok2023discrete} focuses on the lane merging problem and provides safety and recursive feasibility guarantees, but there are no theoretical guarantees on reachability or feasibility enhancement. Additionally, it requires tuning the parameter $\gamma_k$, as it has been shown to significantly impact performance. The proposed approach does not require tuning parameters.

\begin{figure*}[hbt!]
\centering
\subfigure[Case 1: NMPC-DCBF]{
    \includegraphics[width=0.45\textwidth]{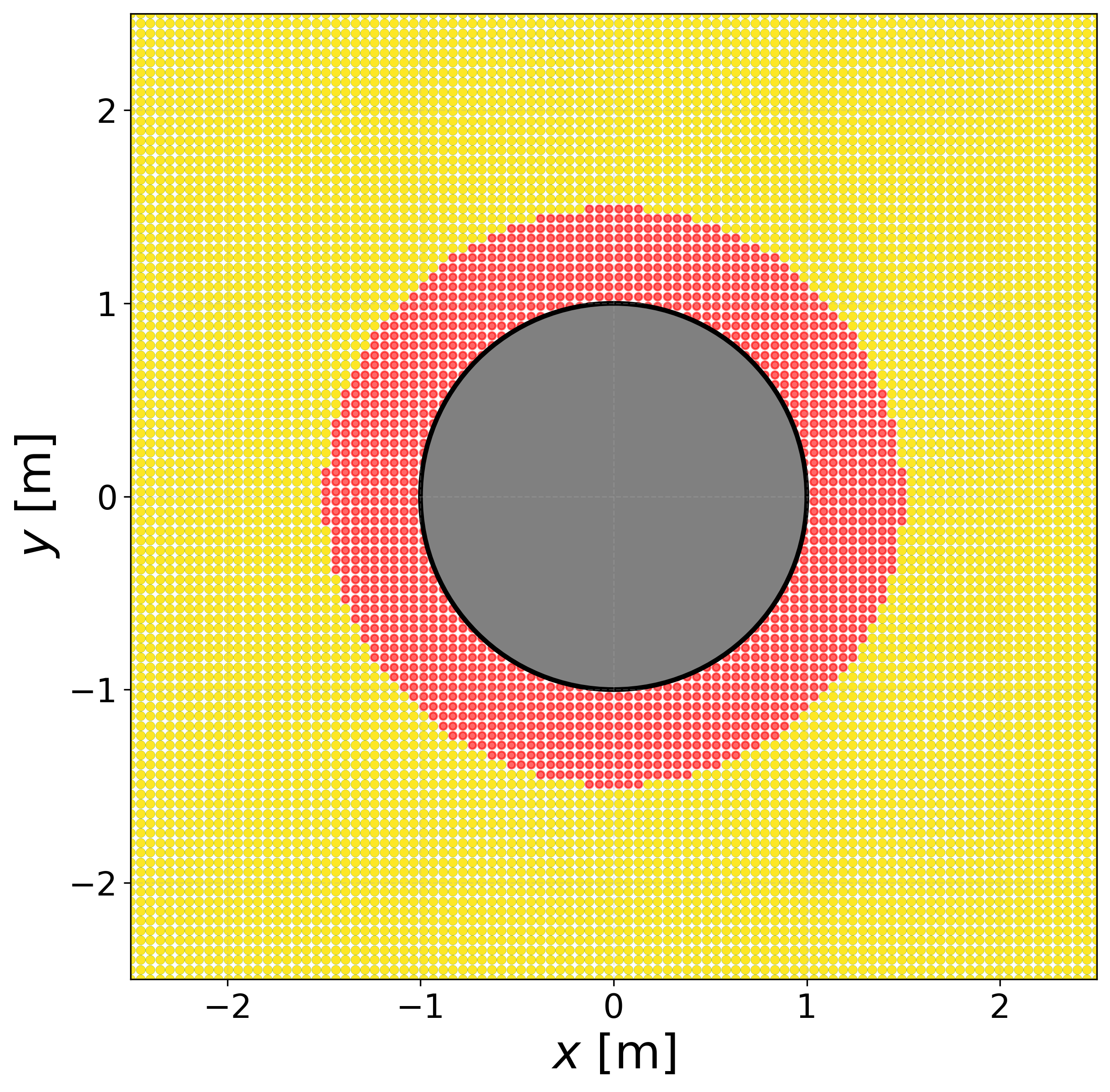}
    \label{subfig:1n}
}
\subfigure[Case 1: MPC-MCI \textbf{(Ours)}]{
    \includegraphics[width=0.45\textwidth]{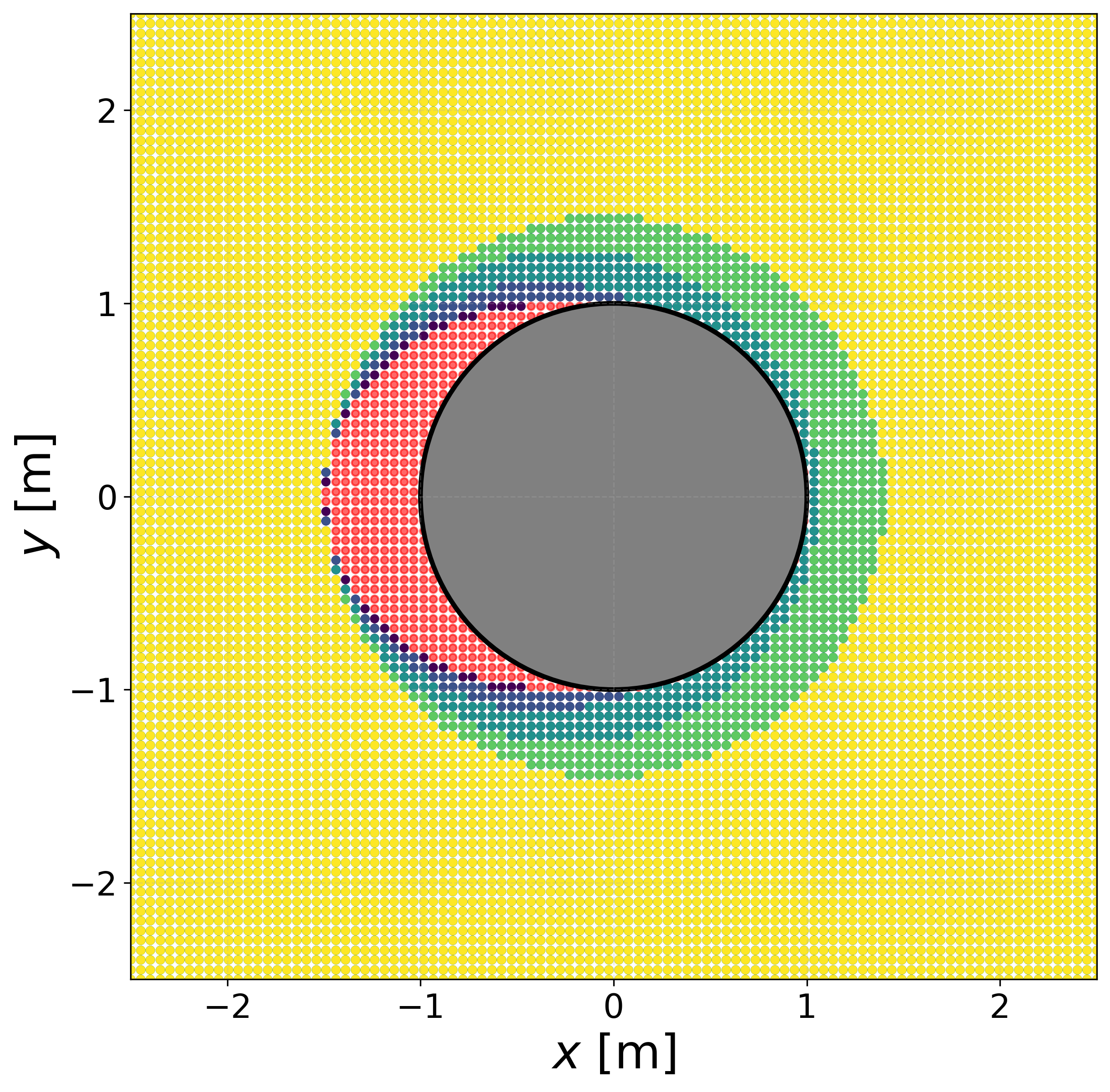}
    \label{subfig:1m}
}
\subfigure[Case 2: NMPC-DCBF]{
    \includegraphics[width=0.45\textwidth]{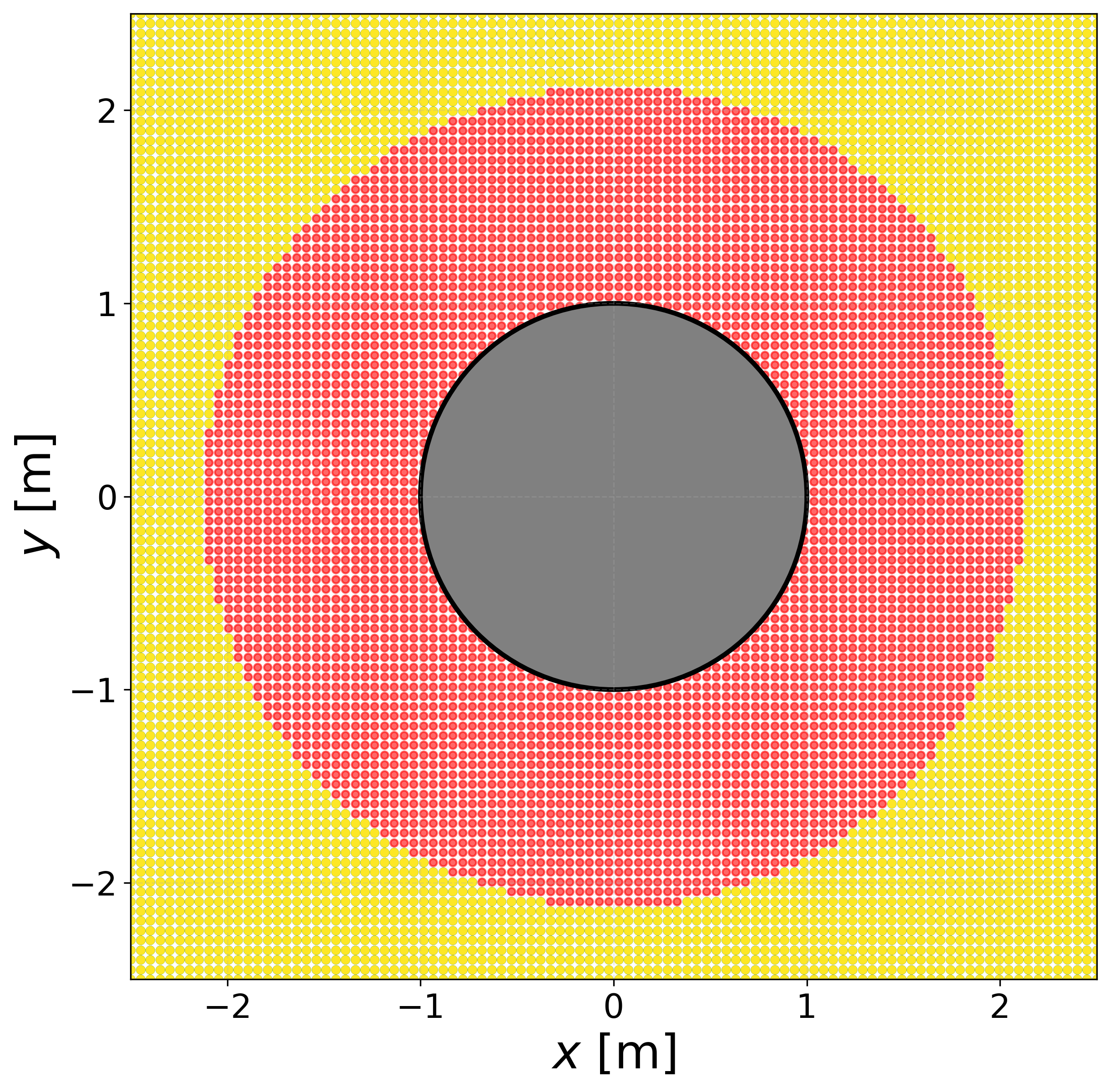}
    \label{subfig:2n}
}
\subfigure[Case 2: MPC-MCI \textbf{(Ours)}]{
    \includegraphics[width=0.45\textwidth]{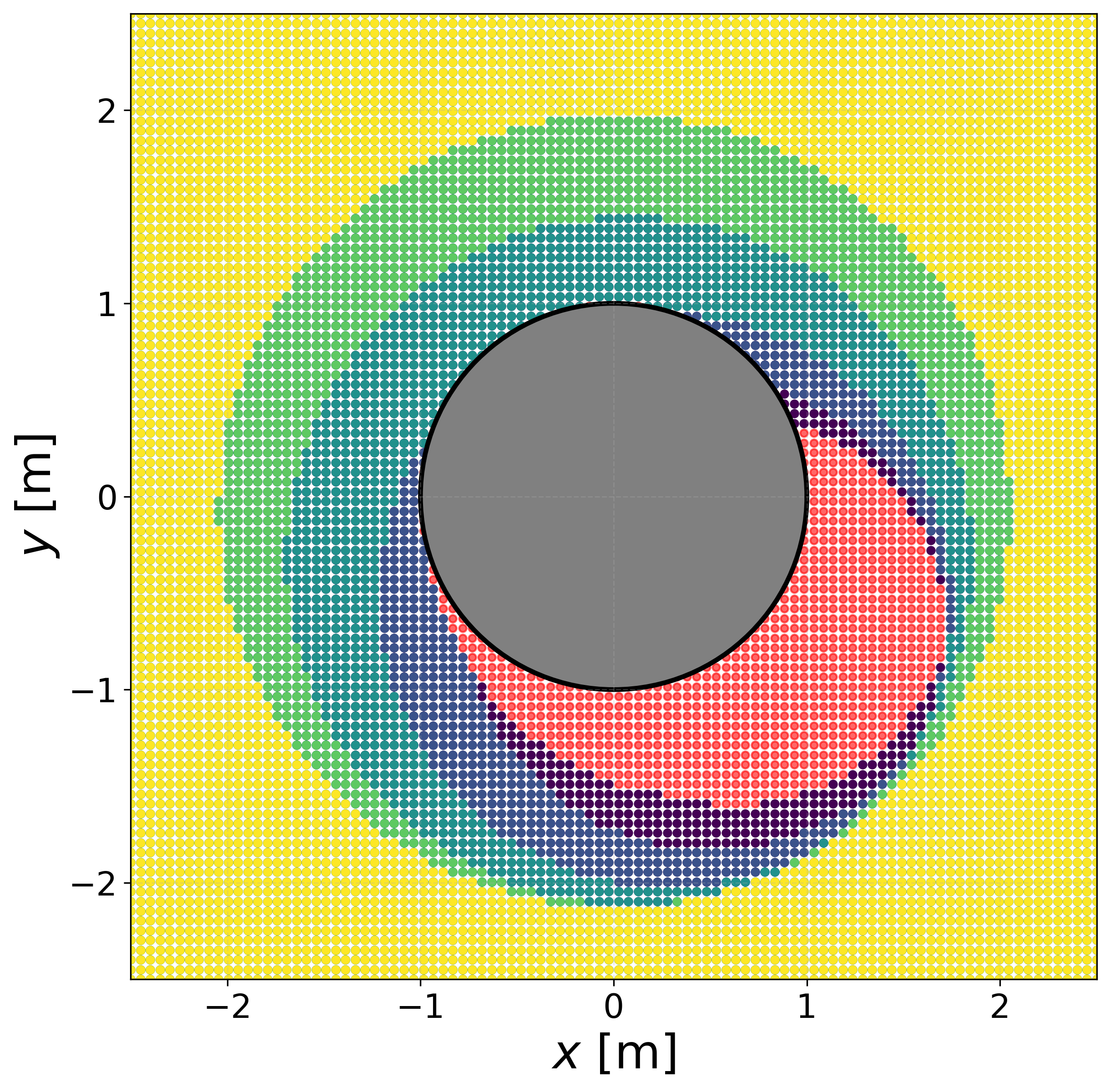}
    \label{subfig:2m}
}
\subfigure{
    \includegraphics[width=0.078\textwidth]{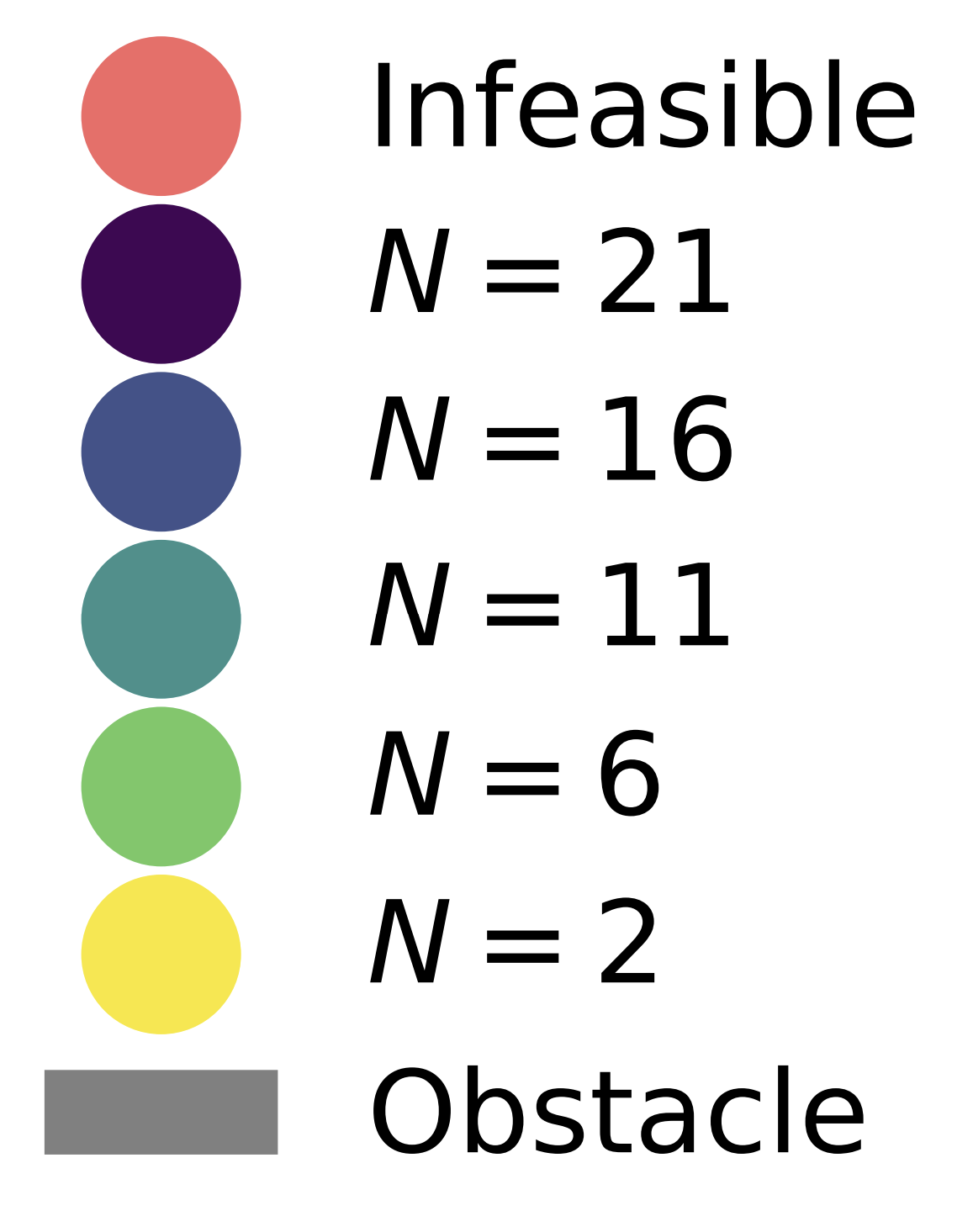}
}
\caption{New feasible region added with an increase of prediction horizon for different positions sampled for NMPC-DCBF and MPC-MCI for different prediction horizons $N=2,6,11,16,21$. In Case 1: $v=\qty{1}{\meter\per\second},\, \theta = \qty{0}{\radian},\, \omega = \qty{0}{\radian\per\second}$ and in Case 2: $v=\qty{1.5}{\meter\per\second},\, \theta = \qty{1.57}{\radian},\, \omega = \qty{2}{\radian\per\second}$. Increasing the prediction horizon adds a new feasible region of MPC-MCI, which is not the case with NMPC-DCBF.}
\label{fig:feas_comparison}
\end{figure*}

\section{Main Result: MPC Formulation Using Distance and Barrier Functions}\label{sec:main}

We propose an MPC formulation that combines a signed distance function $\hdc$ and a discrete-time control barrier function (CBF) $\hcbf$.
The distance constraint enforces that all transient (predicted) states remain outside the unsafe set $\Xa$, while a terminal CBF constraint ensures the terminal state lies in the safe, controlled-invariant set $\Xs$.
This structure yields a tractable controller with strong safety guarantees and reduced conservatism.

\begin{problem}[\textbf{MPC-MCI}]
\decSpace
\label{prob:MPC-CBF-DC}
\begin{align}
    J_{t}^{*}(\X_t){=}\min_{U, X}\ & p(\X_{t+N|t}){+}\sum_{k=0}^{N-1}q(\X_{t+k|t},\U_{t+k|t}) \label{subeq:star-cost}\\
    \operatorname{s.t.}\quad &   \eqref{eq:mpc-system-dynamics},~\eqref{eq:mpc-state-input-constraint},~\eqref{eq:mpc-current-state},~\eqref{eq:mpc-terminal-set}, \notag\\
     &\hdc(\X_{t+k|t}) \geq 0,\ k \in \itt{N-1}  \label{subeq:MPC-CBF-DC-dc}\\
     &\hcbf(\X_{t+N|t}) \geq 0. \label{subeq:MPC-CBF-DC-cbf}
\end{align}
\end{problem}

\noindent
In the above formulation, the constraint~\eqref{subeq:MPC-CBF-DC-dc} forbids predicted collisions along the horizon by requiring non-negative signed distance at each step.
Constraint~\eqref{subeq:MPC-CBF-DC-cbf} anchors the terminal state in $\Xs$, a controlled-invariant safe set (the $0$-superlevel of $\hcbf$), thereby enabling an infinite-horizon safety argument via invariance.
Under Assumption~\ref{a:model} and the compatibility condition
$\hcbf(\X)\ge 0  \Longrightarrow  \hdc(\X)\ge 0$
which holds when $\hcbf$ and $\hdc$ encode the same obstacle geometry (see discussion around Eq.~\eqref{eq:CBFtoDC}), we prove the feasibility and safety in the following two lemmas.

\begin{lemma}[Initial and Recursive Feasibility]\label{thm:feas}
Consider Eq.~\eqref{eq:dis_dynamics} under Assumptions~\ref{a:model} and~\ref{a:state-info}, with $\hdc$ and $\hcbf$ constructed for the same $\Xs$ configuration and satisfying Eq.~\eqref{eq:CBFtoDC}. If $\X_0 \in \Xs$, then Problem~\ref{prob:MPC-CBF-DC} is feasible at $t=0$ and remains feasible for all $t \ge 0$.
\end{lemma}

\begin{proof}
\textbf{Case 1 (Initial feasibility at $t=0$).}
Given $\X_0 \in \Xs$, which is a control invariant set, following from Definition~\ref{def:cnt-inv}, there exists a control trajectory
$\pi:\itt{N-1} \to\mathcal U$
such that
$\phi_f^\pi(t,\X_0) \in \Xs, \quad \forall t \in \itt{N-1}.$
Following Eq.~\eqref{eq:CBFtoDC} and the fact that $\Xs$ is the $0$-superlevel set of $\hcbf$, the control sequence $U = \pi$ satisfies the constraints~\eqref{subeq:MPC-CBF-DC-dc} and~\eqref{subeq:MPC-CBF-DC-cbf}.
Hence, Problem~\ref{prob:MPC-CBF-DC} is feasible at $t=0$.

\textbf{Case 2 (Recursive feasibility: $t \mapsto t+1$).}
Suppose the optimal control trajectory computed at time step $t^*$ is
$\pi^* = \U_{t^*:t^*+N|t^*},$
then the terminal state
$\X_f := \phi_f^{\pi^*}(N, \X_{t^*})$
satisfies $\X_f \in \Xs$, since $\pi^*$ enforces the terminal CBF constraint~\eqref{subeq:MPC-CBF-DC-cbf}.
As $\Xs$ is controlled invariant, there exists $\hat\U \in \mathcal U$ such that
$f(\X_f,\hat\U) \in \Xs.$
Now, define the shifted control trajectory
\begin{equation}
    \hat\pi(t) :=
    \begin{cases}
        \pi^*(t+1), & t \in \itt{N-2}, \\
        \hat\U, & t = N-1.
    \end{cases}
\end{equation}
By construction, $\hat\pi$ satisfies the constraints~\eqref{subeq:MPC-CBF-DC-dc} and~\eqref{subeq:MPC-CBF-DC-cbf} at time $t^*+1$.
Therefore, if Problem~\ref{prob:MPC-CBF-DC} is feasible at $t=t^*$, it is also feasible at $t=t^*+1$.

Combining Case 1 and Case 2, and following the principle of mathematical induction, we conclude that Problem~\ref{prob:MPC-CBF-DC} is feasible for all $t \ge 0$.
\end{proof}

\begin{corollary}
    If Problem~\ref{prob:MPC-CBF-DC} is feasible at $t=0$, then it is feasible for all $t>0$.
\end{corollary}

\begin{lemma}[Safety Guarantee]
Under the conditions of Lemma~\ref{thm:feas}, the closed-loop trajectory satisfies
\[
\X_t \in \Xss,\quad \forall\ t \in \ittinfty,
\]
where $\Xss$ is the maximal controlled invariant subset of $\mathcal X \setminus \Xa$.
\end{lemma}

\begin{proof}
    As Lemma~\ref{thm:feas} holds, let $\pi:\itt{N-1} \to \mathcal U$ denote the control trajectory computed at time step $t$ for the state $\X_t$. Then, as Eq.~\eqref{subeq:MPC-CBF-DC-cbf} is satisfied, therefore $\X_f := \phi^\pi_f(N-1,\X_t) \in \Xs$. Following Definition~\ref{def:cnt-inv}, one can construct a control trajectory $\pi^*: \ittinfty \to\mathcal U$ such that:
    \begin{equation}
        \phi_f^{\pi^*}(t, \X_f) \in \Xs, \quad \forall\ t \in \ittinfty.
    \end{equation}
    Therefore, consider the control trajectory given by:
    \begin{equation}
        \hat \pi(t) := \begin{cases}
            \pi(t), & t \in \itt{N-1}  \\
            \pi^*(t-N), & t\in \{N, N+1, \ldots, \infty \}
        \end{cases}
    \end{equation}
    Under $\hat\pi$, using Eq.~\eqref{eq:CBFtoDC} and the constraint~\eqref{subeq:MPC-CBF-DC-dc} holds for $k\in\itt{N-1}$, the following holds:
    \begin{equation}
        \phi_f^{\hat\pi}(t, \X_t) \in \mathcal X \setminus \Xa, \quad \forall\ t \in \ittinfty
    \end{equation}
    Hence, following the definition of maximal controlled invariant set, $\X_t \in \Xss$.
\end{proof}

\subsection{Theoretical Comparison of Reachability and Feasibility}
In this section, we present the enhancement of feasibility using reachability analysis of the proposed MPC-MCI method compared to NMPC-DCBF, as shown in Ref.~\cite{zeng2021enhancing}, which improves feasibility over MPC-DCBF.
We define one-step reachable states for NMPC-DCBF as follows:
\begin{equation}\label{eq:reachnmpc}
\begin{split}
        \Rnmpc_{N, M}(\X_0) := \{ \X_{1} \in \mathcal X\ |\ \forall i \in \itt {N-1}, \\
        \exists \U_{i} \in \mathcal{U}, \X_i \in \mathcal X, \\
        \X_{i+1} = f(\X_{i}, \U_{i}), \\
        \hcbf (\X_{j+1}) \ge 0, \forall j \in \itt{M-1}\}
\end{split}
\end{equation}

Here, $N\ge1$ is the prediction horizon and $M\ge1$ is the number of steps with the CBF constraint.
This set represents all possible states the system can be in at the next step if it is currently at $\X_0$, as it obtained with the slack variables being $\omega_k = 0$ in constraint~\eqref{subeq:cbf-nmpc-cbf-constraint}.
Similarly, we define the set of one-step reachable states for MPC-MCI:
\begin{equation}\label{eq:reachour}
\begin{split}
        \Roura_N(\X_0) := \{ \X_{1} \in \mathcal X\ |\ \forall i \in \itt {N-1},\\
        \exists \U_{i} \in \mathcal{U}, \X_i \in \mathcal X,\\
        \X_{i+1} = f(\X_{i}, \U_{i}), \hdc(\X_i) \ge 0,\\
    \hcbf (\X_N)\ge0\}
\end{split}
\end{equation}
The key difference in the reachable set definitions is that for NMPC-DCBF, the CBF constraint is imposed for the first $M$ steps; however, for MPC-MCI, the distance constraint is imposed for the first $N-1$ steps, and the CBF constraint is imposed for the terminal state.

\begin{theorem}[Reachability Enhancement]\label{thm:feas_enh}
    With the definitions in Eqs.~\eqref{eq:reachnmpc} and~\eqref{eq:reachour}, we have
    \begin{subequations}
    \begin{align}
        \Rnmpc_{N, M} &= \Rnmpc_{1,1} \label{subeq:FE:NM-11}\\
        \Rnmpc_{1,1} &= \Roura_{1} \label{subeq:FE:11-1}\\
        \Roura_N &\subseteq \Roura_{N+1} \label{subeq:FE:N-N+}
    \end{align}
    \end{subequations}
\end{theorem}

\begin{proof}
    We begin by showing Eq.~\eqref{subeq:FE:NM-11}. Following the definition, we have:
    \begin{equation}
    \begin{split}
        \Rnmpc_{1,1} = \{ \X_{1} \in \mathcal X\ | \
        \exists \U_{0} \in \mathcal{U},\\ \X_{1} = f(\X_{0}, \U_{0}),
    \hcbf (\X_1)\ge0\}.
    \end{split}
    \end{equation}
    As $\Rnmpc_{N,M}$ has a larger number of constraints which are identical for the first state, it follows that $\Rnmpc_{N,M} \subseteq \Rnmpc_{1,1}$.

    To show equality, suppose $\X^*\in\Rnmpc_{1,1}(\X_0)$, then as $\hcbf(\X^*)\ge0$, following Definition~\ref{def:cnt-inv}, one can construct a control trajectory $\pi:\itt \infty \to \mathcal U$ of arbitrary length such that $\hcbf(\phi_f^\pi(t,\X^*)) \ge 0$. Therefore, $\X^*\in\Rnmpc_{N,M}(\X_0)$. Hence,
    \begin{equation}
        \Rnmpc_{N, M} = \Rnmpc_{1,1}. \notag
    \end{equation}
    Similarly, for $N=1$,
    \begin{equation}
    \begin{split}
            \Roura_1(\X_0) := \{ \X_{1} \in \mathcal X\ |\ \exists \U_{0} \in \mathcal{U},\\
            \X_{1} = f(\X_{0}, \U_{0}), \hdc(\X_1) \ge 0, \hcbf (\X_1)\ge0\}
    \end{split}
    \end{equation}
    As $\hcbf\ge0 \implies \hdc\ge0$, hence $\Roura_1=\Rnmpc_1$.

    Moving on to Eq.~\eqref{subeq:FE:N-N+}, suppose $\X^* \in \Roura_{N}(\X_0)$, then let $\pi:\itt{N-1} \to \mathcal{U}$ be the corresponding control trajectory and $\X_N^*$ be the corresponding terminal state. Then, as $\hcbf(\X_N^*)\ge0$, there exists $\hat\U$ such that $\hcbf(f(\X_N^*,\hat\U)) \ge 0$ and $\hdc(\X_N^*)\ge 0.$ Therefore, under the control trajectory defined as
    \begin{equation}
        \hat \pi (t) :=\begin{cases}\pi(t), & t\in \itt{N-1}\\
        \hat \U, & t = N
        \end{cases},
    \end{equation}
    the state $\X^* \in \Roura_{N+1}(\X_0)$. Hence, Eq.~\eqref{subeq:FE:N-N+} holds.
\end{proof}

\begin{rmrk}
    Following this analysis, the proposed method has a larger one-step reachable set than NMPC-DCBF. Furthermore, in simulations, we show that as the prediction horizon increases, our method enlarges the reachable set, whereas NMPC-DCBF does not.
\end{rmrk}

If the reachable set from a given state is non-empty, this indicates that the state satisfies the initial feasibility condition and ensures safety for all future time steps. Furthermore, since the one-step reachable sets expand monotonically with respect to set inclusion as the prediction horizon $N$ increases (as shown in Eq.~\eqref{subeq:FE:N-N+}), the corresponding sequence of feasible state sets also forms a non-decreasing sequence under set inclusion.

\begin{figure}[hbt!]
    \centering
    \includegraphics[width=0.80\textwidth]{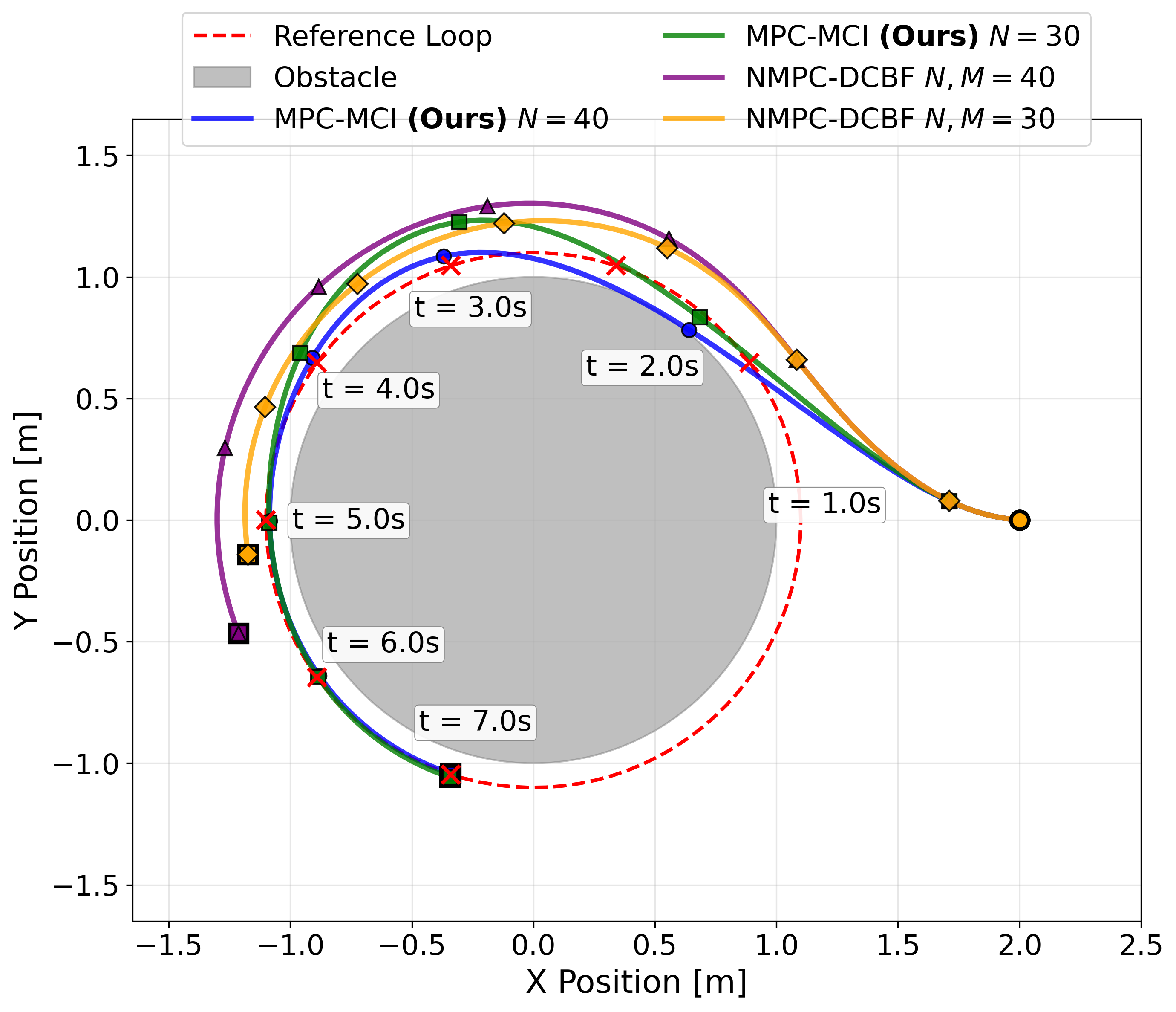}
    \caption{Trajectory tracking for the unicycle for different algorithms. The markers on the trajectories are equally separated in time. MPC-MCI is able to track a high-speed trajectory close to the obstacle, while ensuring safety.}
    \label{fig:traj}
\end{figure}

\section{Simulation Results}\label{sec:sim}
In this section, simulations are carried out to compare with other methods and demonstrate the advantages of our proposed method in terms of feasibility and reachability. Consider a discrete-time unicycle system with states $\X = [x,y,\theta,v,\omega]^\top \in \R^5$ and the acceleration control inputs $\U = [a, \alpha]^\top \in [-a_m, a_m]\times[-\alpha_m, \alpha_m]$, where $a_m>0$ and $\alpha_m>0$ are the control bounds.
The discretized dynamics are represented as:
\begin{equation}
        \X_{n+1} = f(\X_n, \U_n) = \begin{bmatrix}
        x+\cos\theta (v\Delta t+\frac 12 a\Delta t^2)\\
        y + \sin\theta (v\Delta t+\frac 12 a\Delta t^2)\\
        \theta +\omega \Delta t \\
        v+a\Delta t\\
        \omega + \alpha \Delta t
    \end{bmatrix}
\end{equation}

For numerical simulations, the sampling time is set as $\Delta t =\qty{0.05}{\second}$ and the control bounds are set as $a_m = \qty{1}{\meter\per\second\squared}, \alpha_m = \qty{1}{\radian\per\second\squared}$. All simulations are run in Python, and the optimization problem is formulated using CasADi v3.7.0~\cite{andersson2019casadi} and solved using IPOPT.

\subsection{Obstacle Configuration and Control Barrier Function}
A circular obstacle of radius $r_0=\qty{1}{\meter}$ is assumed to be present at the origin $(0,0)$. Therefore,
\begin{equation}
    \hdc(\X) = \sqrt{x^2+y^2} -r_0
\end{equation}
We propose that the following function is a control barrier function for unicycle dynamics:
\begin{equation}
    \hcbf(\X) = \sqrt{x^2+y^2} - r_0 - \frac{v^2}{2a_m}
\end{equation}
It can be shown that under the following recovery control law $\pi_s:\mathcal X \to \mathcal U$ with $\pi_s(\X) = [a_s, \alpha_s]^\top$:
\begin{equation}
    \alpha_s =0, \quad a_s = \begin{cases}
        -a_m, & v>a_m \Delta t \\
        -\frac{2v}{\Delta t},& -\frac12a_m\Delta t\le  v \le\frac12a_m \Delta t\\
        a_m, & v<-a_m\Delta t
    \end{cases}
\end{equation}
the function $\hcbf$ is a control barrier function as $\hcbf (f(\X, \pi_s(\X))) \ge \hcbf (\X)$.
The optimizer was observed to face issues with $\sqrt{\cdot}$ as constraints, hence the following functions were used, which have the same $0$-superlevel set, i.e., $d\ge0 \iff d'\ge 0$ and $h\ge0 \iff h'\ge 0$:
\begin{equation}
    \hdc'(\X) = {x^2+y^2} -r_0^2,\quad \hcbf'(\X) = {x^2+y^2} - \left(r_0 + \frac{v^2}{2a_m}\right)^2 \label{eq:sim_cbf}
\end{equation}

\subsection{Numerical Results for Feasibility Enhancements}
To demonstrate feasibility enhancement, we sample states and verify the numerical feasibility of the optimization problem in each approach. Since the state space has five dimensions, we fix $\theta, v, \omega$ to specific values and check feasibility for a grid of positions around the obstacle. The states are sampled from:
\begin{equation}
\begin{split}
    \mathcal X_{\theta_0, v_0, \omega_0} = \{(x,y,\theta,v,\omega) \in \R^5\ |\ x_{\min}\le x \le x_{\max},\\ y_{\min} \le y \le y_{\max},  v = v_0, \theta = \theta_0, \omega = \omega_0\}.
\end{split}
\end{equation}
For the simulation, we keep $x_{\min} , y_{\min} = \qty{-2.5}{\meter}, x_{\max} , y_{\max} = \qty{2.5}{\meter}$.
As demonstrated in Ref.~\cite{zeng2021enhancing}, the NMPC-DCBF framework improves feasibility compared to earlier methods. Consequently, we limit our comparative analysis to this algorithm. Additionally, a numerical example illustrating the improvement in feasibility achieved by DTCBF-MPC is presented in the appendix. As feasibility performance does not depend on the cost functions, we set the slack variables $\omega_k = 0$ in NMPC-DCBF. To present a comprehensive analysis of performance, we vary the prediction horizon $N$ and use the same control barrier function constraint on $\hcbf'$ as defined in Eq.~\eqref{eq:sim_cbf}. As shown in Theorem~\ref{thm:feas_enh}, the feasible region is independent of $\mcbf$, hence we set it to be $\mcbf=N$.

For case 1, we set $v=\qty{1}{\meter\per\second},\, \theta = \qty{0}{\radian},\, \omega = \qty{0}{\radian\per\second}$ with the plots for NMPC-DCBF in Fig.~\ref{subfig:1n} and for MPC-MCI in Fig.~\ref{subfig:1m}. The limiting teardrop shape of the infeasible region for MPC-MCI demarcates the true unsafe region, as a collision is bound to occur when speed is high, the vehicle is oriented towards, and close to the obstacle. Therefore, they are unsafe and unfeasible.

For case 2, we set $v=\qty{1.5}{\meter\per\second},\, \theta = \qty{1.57}{\radian},\, \omega = \qty{2}{\radian\per\second}$ with the plots for NMPC-DCBF in Fig.~\ref{subfig:2n} and for MPC-MCI in Fig.~\ref{subfig:2m}. Here, as the unicycle is oriented towards the $+y$-axis, we expect a teardrop oriented with the $y$-axis; however, due to a high $\omega$, the unicycle is bound to collide in a teardrop smeared to the right, as the uncontrolled trajectory will be a circular path curving to the left. The infeasible region is larger than in case 1, as the linear speed is higher.

The fraction of feasible samples for the two cases is also presented in Table~\ref{tab:feas_comp}, and the trend is clear: for a larger prediction horizon, the feasible region increases for MPC-MCI, whereas it remains unchanged for NMPC-DCBF. We also note that the feasible fraction for any prediction horizon NMPC-DCBF is strictly smaller than the feasible region for MPC-MCI with $N=2$. The number of points sampled for each case is $10{,}000$, and for case 1, the number of infeasible points reduces from $2756$ to $1582$, and for case 2, they reduce from $5560$ to $2034$. Therefore, we observe a reduction by a factor of 1.74 and 2.73.
The system enters the region $\hcbf<0$, while ensuring safety, thereby reducing the conservatism imposed by the control barrier function.
This demonstrates the importance of setting the CBF constraint only at the terminal state and having a function $d$ that denotes the known unsafe region.

\begin{table}[hbt!]
\centering
\caption{Fraction of Feasible States for different prediction horizon $N$.}
\begin{tabular}{lccc}
\hline
\multirow{2}{*}{Method} & \multirow{2}{*}{$N$} & \multicolumn{2}{c}{Feasible Fraction} \\
\cline{3-4}
                        &                    & Case 1 & Case 2 \\
\hline
\multirow{5}{*}{NMPC-DCBF}
                        & 2  & 0.7244 & 0.4440 \\
                        & 6  & 0.7244 & 0.4440 \\
                        & 11 & 0.7244 & 0.4440 \\
                        & 16 & 0.7244 & 0.4440 \\
                        & 21 & 0.7244 & 0.4440 \\
\hline
\multirow{5}{*}{\textbf{MPC-MCI (Ours)}}
                        & 2  & 0.7404 & 0.4822 \\
                        & 6  & 0.7980 & 0.6103 \\
                        & 11 & 0.8298 & 0.7245 \\
                        & 16 & 0.8386 & 0.7743 \\
                        & 21 & 0.8418 & 0.7966 \\
\hline
\end{tabular}
\label{tab:feas_comp}
\end{table}

\subsection{Numerical Results for Reachability Enhancement}
To demonstrate the benefits of a larger reachable region, we consider the problem of trajectory tracking.
The reference trajectory is defined as:
\begin{equation}
    (x_r(t), y_r(t)) = \left(r_r \cos\left(\frac{2\pi}{t_r} t\right),r_r\sin\left(\frac{2\pi}{t_r} t\right)\right)
\end{equation}
which represents a loop of radius $r_r = \qty{1.1}{\meter}$ and a loop time $t_r = \qty{10}{\second}$. To track this trajectory, $v=\qty{0.69}{\meter\per\second}$ and $\sqrt{x^2+y^2} =r_r$. Therefore, the CBF value is $h = r_r-r_0-v^2/2a_m = -0.14$. Hence, this trajectory is inside the unsafe region of the CBF.
The cost function is defined as a sum of a quadratic function of position error and a quadratic function of control input, independent of other states.

For comparison, NMPC-DCBF is run for $N=\mcbf=30,\, 40$, and the proposed approach, MPC-MCI, is run for $N=30,\,40$. It was observed that a longer prediction horizon was needed for the trajectories to converge to the reference trajectory, due to the simple cost function and the nonholonomic dynamics.

To warm-start the simulation, we use the previously predicted trajectory, appended with the recovery controller given in $\pi_s$ for the last predicted state. This leads to a significant speed-up in simulations and ensures that the optimizer starts with a feasible guess.

The results in Fig.~\ref{fig:traj} demonstrate that MPC-MCI is able to converge and track the desired trajectory. However, due to conservatism, NMPC-DCBF struggles to do so, either remaining far from the trajectory or moving too slowly to follow the reference. This demonstrates that the proposed method has a larger reachable set, allowing for tracking trajectories that are not possible with only CBF constraints, as in NMPC-DCBF.

\section{Conclusion}\label{sec:conc}
In this paper, we have shown that imposing a CBF as a single terminal constraint, coupled with a signed-distance constraint along the prediction horizon, yields a provably safe controller that admits trajectories venturing into regions of unknown safety, $\Xss \setminus \Xs$, that earlier formulations exclude by construction. Unlike prior approaches, the proposed MPC-MCI does not rely on slack variables, decay-rate tuning, two-step quasi-CBF constructions, or auxiliary constraint functions to enhance feasibility, and yet it strictly outperforms them on the same problem instances. The constructive nature of the recursive feasibility proof further enables warm-starting the nonlinear optimization with the previously predicted trajectory appended by the recovery controller, which substantially reduces computational effort while preserving the feasibility guarantee.

Theoretically, the formulation provides formal safety, recursive feasibility, and a monotone reachability enhancement with respect to the prediction horizon, all in a general nonlinear setting that is not restricted to linear or specialized dynamics. The numerical study on a nonholonomic unicycle corroborates this analysis: the infeasible region shrinks by a factor of $1.74$ to $2.73$ across the tested cases, the feasible set grows monotonically with the horizon, and the closed-loop system successfully tracks reference trajectories that lie entirely inside the unsafe region of the CBF, which NMPC-DCBF cannot follow due to its inherent conservatism. Together, these results indicate that anchoring safety at the terminal state, rather than along the entire horizon, is a principled way to recover the optimality lost to conservative CBF design without compromising on safety guarantees.

\section*{Appendix: Numerical Comparison with DTCBF-MPC and MPC-MCI}\label{appen:com}

Consider a 1D double integrator with sampling time $\Delta t=1$:
\begin{equation*}
\begin{split}
x_{t+1} = x_t + v_t + \tfrac{1}{2}u_t,
& \quad
v_{t+1} = v_t + u_t, \\
 |u_t|\le u_{\max},
&\quad |v_t| \le \tfrac{1}{2}u_{\max}.
\end{split}
\end{equation*}
The unsafe set is $\mathcal{X}_a = \{x<0\}$, so safety requires $x \ge 0, \forall t \ge 0$.
We define $\hdc(\X)=x$ and a conservative CBF $\hcbf(\X)=x-\alpha$ so that $\Xs=\{x\ge \alpha\}$. It can be shown that $\hcbf(\X)$ is a valid CBF under the state constraints for $\alpha > 0$.

\paragraph{Our approach (distance each step + terminal CBF).}
For horizon $N=2$, the constraints are:
\begin{align*}
d(\X_{1|0})\ge0 &\implies x_1 = x_0 + v_0 + \tfrac{1}{2}u_0\ge 0, \\
h(\X_{2|0})\ge0 &\implies x_2 = x_1 + v_1 + \tfrac{1}{2}u_1\ge \alpha,
\end{align*}
i.e., a distance constraint on the transient state and a terminal CBF constraint.

\paragraph{CBF at the last two steps.}
Enforcing discrete-time CBF at $k=0,1$ with decay rate $\gamma\in(0,1]$ gives
\begin{align*}
    h(\X_{1|0})\ge0 &\implies  x_1 = x_0 + v_0 + \tfrac{1}{2}u_0\ge \alpha,\\
    h(\X_{2|0})\ge(1-\gamma)h(\X_{1|0}) &\implies x_2-\alpha\ge (1-\gamma)(x_1-\alpha)\\
    & \implies v_1+\tfrac12 u_1\ge -\gamma(x_1-\alpha)
\end{align*}

\paragraph{Numerical illustration.}
Let
\begin{equation*}
x_0=0.1,\quad v_0=-0.7,\quad u_{\max}=1.5, \quad \alpha = 0.2
\end{equation*}

\begin{itemize}
    \item \textbf{Our approach:}
    At step $0$, we set $u_1 = u_{\max}$ and $u_2 = 0$, therefore, we obtain $x_1 = 0.15$ and $x_2 = 0.95$. It satisfies the constraints as $d(\X_{1|0}) = 0.15$ and $h(\X_{2|0}) = 0.75$. Thus, the problem is feasible.
    \item \textbf{CBF at last two steps:}
    At step $0$, the CBF condition requires
    \[
    \tfrac12 u_0 \ge 0.8
    \]
    which exceeds $u_{\max}=1.5$.
    Hence, the problem is infeasible \emph{at the very first step}, even though a safe trajectory exists.
\end{itemize}

Our formulation is less conservative: it admits feasible, safe trajectories for which enforcing the CBF at the last two steps renders the problem infeasible.

\bibliography{main}

\end{document}